\documentclass[11pt]{article}

\pdfoutput=1

\usepackage[utf8]{inputenc}
\usepackage[T1]{fontenc}
\usepackage[english]{babel}
\usepackage{lmodern}
\usepackage{microtype}
\usepackage{amsmath,amssymb,amsthm}
\usepackage{mathtools}
\usepackage{algorithm}
\usepackage{algorithmic}
\usepackage{enumitem}
\usepackage{xcolor}
\usepackage[hidelinks]{hyperref}
\usepackage[capitalize,nameinlink]{cleveref}
\usepackage[a4paper,margin=1in]{geometry}

\newcommand{\ZZ}{\mathbb{Z}}
\newcommand{\A}{{\tilde \Psi}}

\theoremstyle{plain}
\newtheorem{theorem}{Theorem}[section]
\newtheorem{lemma}[theorem]{Lemma}

\theoremstyle{definition}

\theoremstyle{remark}

\title{Sublinear Time Algorithms for Abelian Group Property Testing}
\author{Nader H. Bshouty\\Technion -- Israel Institute of Technology}
\date{}

\begin{document}

\maketitle

\begin{abstract}
In this paper, we study the problems of abelian group property testing in two models. In the {\it partially specified model} (PS-model), the algorithm does not know the group size but can access randomly chosen elements of the group, along with the Cayley table of these elements, which provides the result of the binary operation for every pair of selected elements. In the stronger {\it fully specified model} (FS-model), the algorithm knows the size of the group and has access to all its elements and the Cayley table. 

In property testing of abelian group property, given a finite set $G$ and oracle access to a binary operation $*:G^2\to G$, we aim to distinguish whether $(G,*)$ is an abelian group or is $\epsilon$-far from any abelian group over~$G$.

Using a novel approach, we present a tester in the PS-model (and consequently in the FS-model) that runs in time $\tilde O(\sqrt{|G|}+1/\epsilon)$, improving upon the Goldreich-Tauber tester, which runs in time $O(|G|/\epsilon)$. Additionally, our tester improves another tester by Goldreich and Tauber that runs in time $O(|G|^2)$ and makes $\tilde O(|G|+1/\epsilon)$ queries. 

We further extend our result to testing subclasses of abelian groups ${\cal G}$ that are closed under isomorphism. Specifically, if one can decide in time $T$ whether an abelian group of the form $\ZZ_{m_1}\times \cdots\times \ZZ_{m_r}$ belongs to ${\cal G}$, then there exists a tester for ${\cal G}$ that runs in time $T+\tilde O(\sqrt{|G|}+1/\epsilon)$ and makes $O(\sqrt{|G|}+1/\epsilon)$ queries. 
This result gives testers that run in time $O(\sqrt{|G|}+1/\epsilon)$ for subclasses such as abelian groups of rank at most $k$, abelian $p$-groups, and vector spaces over~$\ZZ_p$. 

We then present two subclasses, ${\cal G}_1$ and ${\cal G}_2$, of abelian groups that are closed under isomorphism for which any tester for ${\cal G}_1$ in the FS-model must run in time $\Omega(|G|^{1/4}+1/\epsilon)$, and any tester for ${\cal G}_2$ in the PS-model must run in time $\Omega(\sqrt{|G|}+1/\epsilon)$, showing that our approach provides tight bounds for certain subclasses of abelian groups. 
\end{abstract}

\section{Introduction}
Given the $|G|\times |G|$ multiplication table of a set $G$. Consider the task of deciding whether it represents the Cayley table of an Abelian group. Any algorithm must inspect every entry of the table, implying a lower bound of $\Omega(|G|^2)$. Evra et al.~\cite{EvraGKK24} gave an  $O(|G|^2)$-time algorithm. 

In this paper, we consider the {\it property-testing} version of the problem: given query access to the table, determine whether it is the Cayley table of an Abelian group, or is far from any such table. We show that this can be done in time $O(\sqrt{|G|})$, even when the size of the table is unknown, and the elements of $G$ are provided in a random order.  

Property testing of sets of objects was first considered in the seminal works of Blum, Luby, and Rubinfeld~\cite{BlumLR93} and Rubinfeld and Sudan~\cite{RubinfeldS96}. It has since become a very active area of research; see, for example, the surveys and books~\cite{CzumajS06,GoldreichSurvey10,Goldreich17,Ron08,Ron09}. 

In this paper, we study the problem of testing abelian groups: Given a finite set $G$ and oracle access to a multiplication table $*:G^2\to G$, a tester aims to distinguish whether $(G,*)$ is an abelian group or is $\epsilon$-far from any abelian group over~$G$; that is, for every abelian group $(G,\cdot)$, we have $$|\{(x,y)|x*y\not=x\cdot y\}|\ge \epsilon |G|^2.$$ This is known as the {\it Hamming distance}. The {\it edit distance} has also been considered in the literature~\cite{FriedlIS05,GallY13}. 

In this paper, we study the problems of abelian group property testing in two models. In the {\it fully specified model} (FS-model), the size of the groups is known, and the algorithm has access to the group's elements and their Cayley table. In the {\it partially specified model} (PS-model), the size of the group is unknown; the algorithm can receive uniform random elements of the group and access the Cayley table of elements observed so far. 

In the FS-model, Erg\"un et al.~\cite{ErgunKKRV00} presented a tester for groups that runs in time $\tilde O(|G|^{3/2}/\epsilon)$. Goldreich and Tauber~\cite{OdedT23} presented a tester for groups and abelian groups that runs in time $\tilde O(|G|/\epsilon)$ and makes $\tilde O(|G|/\epsilon)$ queries. They also provided another tester that runs in time $\tilde O(|G|^2)$ and makes $\tilde O(|G|+1/\epsilon)$ queries. 

In this paper, we give the first sublinear-time algorithm (in the size of the group). 
\begin{theorem}
        There exists a tester for abelian groups in the PS-model (and therefore in the FS-model) that runs in time $$\tilde O(\sqrt{|G|}+1/\epsilon).$$ 
\end{theorem}

This is the first sublinear-time tester for the abelian groups.

We then extend this result to testing any subclass of abelian groups that is closed under isomorphism. In the PS-model, we prove the following.
\begin{theorem}
    Let ${\cal G}$ be a subclass of abelian groups that is closed under isomorphism. If there exists an algorithm that can determine, in time $T$, whether any abelian group of the form $\ZZ_{m_1}\times \cdots\times \ZZ_{m_r}$ belongs to ${\cal G}$, then there exists a tester for ${\cal G}$ in the PS-model that runs in time $T+\tilde O(\sqrt{|G|}+1/\epsilon)$ and makes $\tilde O(\sqrt{|G|}+1/\epsilon)$ queries.

    In particular, there exist testers that run in time $\tilde O(\sqrt{|G|}+1/\epsilon)$ for the classes: abelian groups of rank at most \footnote{A group $G$ is abelian of rank at most $k$ if it is isomorphic to $\ZZ_{m_1}\times \cdots\times\ZZ_{m_s}$ where $s\in [k]$ and $m_i\in{\mathbb N}$. It can also be defined as the minimal number of generators required to generate the group.}$k$, abelian $p$-groups, vector spaces over $\ZZ_p$ and cyclic groups.
\end{theorem}
For cyclic groups in the FS-model, Gall and Yoshida~\cite{GallY13} gave a tester that runs in time $poly(\log |G|)/\epsilon$. In the PS-model, they gave a lower bound of $\Omega(|G|^{1/6})$ for testing cyclic groups and linear space over ${\mathbb{Z}}_p$. 

Two lower bounds are known for the number of queries in the FS-model. Goldreich and Tauber~\cite{OdedT23} gave the lower bound $\Omega(\log |G|)$ for testing any subclass of groups ${\cal G}$ that is closed under isomorphism and contains all the cyclic groups of prime order. Le Gall and Yoshida~\cite{GallY13} gave a lower bound $\Omega(|G|^{1/6-4/(6(3k+1))})$ for testing the class of abelian groups of rank at most $k$. 

In this paper, we prove the following lower bound.

\begin{theorem}
    There exists a subclass ${\cal G}$ of abelian groups that is closed under isomorphism, for which any tester for ${\cal G}$ in the FS-model with $\epsilon<1/4$ must make at least $\Omega(|G|^{1/4})$ queries. 
\end{theorem}

Additionally, we establish the following tight lower bound.
\begin{theorem}
    There exists a subclass ${\cal G}$ of abelian groups that is closed under isomorphism, for which any tester for ${\cal G}$ in the PS-model with $\epsilon<1/4$ must make at least $\Omega(\sqrt{|G|})$ queries. 
\end{theorem}

The paper is organized as follows: In Section~\ref{OT}, we introduce the techniques used to prove both upper and lower bounds. Section~\ref{AGPR} presents preliminary results. In Section~\ref{Sec4}, we describe the one-sided error tester, followed by the two-sided error tester in Section~\ref{ATAG}. Section~\ref{TSAG} extends our results to testers for any subclass of abelian groups that is closed under isomorphism. Finally, in Section~\ref{SecLB}, we establish the lower bounds.

\section{Our Technique}\label{OT}
Our approach is novel and does not rely on testing the associativity, commutativity, and cancellativity of $(G,*)$, as previous methods have done~\cite{ErgunKKRV00,OdedT23,Pak12,RajagopalanS00}.

Given a finite {\it magma} $(G,*)$ -- a finite set $G$ with a binary operation $*:G\times G\to G$ -- 
we begin by running the algorithm from~\cite{bshouty2025} that constructs generators for $(G,*)$ with triangular relations\footnote{A set of generators $\{a_1,a_2,\ldots,a_t\}$ of a group $G$ has a {\it triangular relations} if, for each $i\in [t]$, there exists a relation that depends on $a_1,\ldots,a_i$.}. This algorithm runs in time $\tilde O(\sqrt{|G|})$. If the algorithm fails, we reject. If it successfully returns a set of generators with triangular relations, denoted as~${\cal R}$, we use these relations to construct an abelian group $(\Gamma({\cal R}),\cdot)$ that satisfies those relations. 

The key results ensuring the correctness of our approach are as follows:
\begin{enumerate}
    \item If $(G,*)$ is an abelian group, then it is isomorphic to $\Gamma({\cal R})$ via an isomorphism $\Psi:\Gamma({\cal R})\to G$.
    \item For every $x\in\Gamma({\cal R})$, $\Psi(x)$ can be computed in time $poly(\log|G|)$.
    \item Both multiplication and inversion in $\Gamma({\cal R})$ can be computed in time $poly(\log|G|)$.
\end{enumerate}

Before presenting our primary tester, we first introduce a simple tester that runs in time $\tilde O(|G|+1/\epsilon)$. The main idea is to construct, in time $\tilde O(|G|)$, a table of the isomorphic values $\Psi:\Gamma({\cal R})\to G$. Once this table is constructed, we check whether $\Psi$ is a bijective map. If not, then we reject.  

If $\Psi$ is bijective, we then, choose $O(1/\epsilon)$ pairs $(x,y)\in \Gamma({\cal R})^2$ uniformly at random and verify whether $\Psi(x\cdot y)=\Psi(x)*\Psi(y)$. 

We then prove that this algorithm provides a tester that runs in time $\tilde O(|G|+1/\epsilon)$.

To achieve an improved time complexity of $\tilde O(\sqrt{|G|}+1/\epsilon)$, we avoid constructing the entire table of $\Psi$ and, therefore, cannot directly verify whether $\Psi$ is bijective. Instead, we first check whether $\Psi$ is $\Theta(1)$-close to a bijective map and then verify whether $\Psi(x\cdot y)=\Psi(x)*\Psi(y)$. The first step requires $O(\sqrt{|G|})$ queries, and the second step requires $O(1/\epsilon)$ queries. 

Although this two-step approach may not be universally applicable, we show that it holds for abelian groups due to their specific algebraic structure. 

For testing subclasses ${\cal G}$ of abelian groups that are closed under isomorphism, we use the same tester and add a verification procedure to check that $\Gamma({\cal R})$ belongs to ${\cal G}$. 

The lower bound results follow directly from Theorems 7 and 8 in \cite{bshouty2025}. Theorem~6 establishes a lower bound of $\Omega(|G|^{1/4})$ on the number of queries required by any algorithm in the FS-model to distinguish between groups ${\cal G}$ that are isomorphic to either $H_1={\mathbb{Z}}_{p^2}^m$ or $H_2=\mathbb{Z}_{p^2}^{m-1}\times \mathbb{Z}_p^2$. Theorem~7 establishes the lower bound $\Omega(\sqrt{|G|})$ on the number of queries required by any algorithm in the PS-model to distinguish between abelian groups ${\cal G}$ that are isomorphic to either $G_1=\mathbb{Z}_p^m$ or $G_2=\mathbb{Z}_p^{m-1}$. 

\subsection{One-sided vs. Two-sided Error Tester}
A tester is called a {\it one-sided error tester} if it accepts with probability $1$ when $(G,*)$ is an abelian group.

Our algorithm that identifies generators of $G$ is randomized. Therefore, there is a nonzero probability that the tester may fail even when $(G,*)$ is an abelian group. As a result, our tester is a {\it two-sided error} tester. 

In~\cite{bshouty2025}, Bshouty also gives a deterministic algorithm that identifies generators with triangular relations for $(G,*)$ in time $O(|G|)$. Using this algorithm, we obtain a one-sided tester that runs in time $\tilde O(|G|+1/\epsilon)$, improving upon the one-sided tester of Goldreich and Tauber, which runs in time $O(|G|/\epsilon)$.

We further extend the result to testing subclasses of abelian groups that are closed under isomorphism. Specifically, for any subclass of abelian group ${\cal G}$, if there exists an algorithm that can determine in time $T$ whether an abelian group of the form $\ZZ_{m_1}\times \cdots\times \ZZ_{m_r}$ belongs to ${\cal G}$, then there exists a one-sided error tester for ${\cal G}$ that runs in time $T+\tilde O({|G|}+1/\epsilon)$ and makes $O({|G|}+1/\epsilon)$ queries. 

This gives one-sided error testers that run in time $O({|G|}+1/\epsilon)$ for the subclasses such as abelian groups of rank at most $k$, abelian $p$-groups and vector spaces over~$\ZZ_p$. 

\section{Abelian Groups and Preliminary Results}\label{AGPR}
In this section, we provide definitions and preliminary results that will be used to prove the main results.
\subsection{Basic Facts in Groups}
Let $(G,\cdot)$ be an abelian (commutative) group, and let $e$ be its identity element. 
For a nonempty subset $A$ of $G$, we define
$\langle A\rangle:=\{a_1\cdots a_k|a_i\in A,k\in\mathbb{N}\}$
as the subgroup {\it generated} by~$A$. If $\langle A\rangle =G$, we say that $A$ is a {\it set of generators} of $G$. 

A {\it relation} among elements of a group describes an equation that holds between those elements, typically reflecting how generators interact to define the group's structure. We say that a set of generators $A=\{a_1,\ldots,a_m\}$ of a group $G$ satisfies a {\it triangular  relation} if, for every $i$, there exists a relation that depends on $a_i$ and some elements in $\{a_1,a_2,\ldots,a_{i-1}\}$. 
For example, if $A=\{a,b,c\}$ then $a^5=e$, $b^3=a^2$ and $c^7=ab^2$ are triangular relations\footnote{The name ``triangular relations'' originates from triangular matrices, where -- using $+$ and $0$ instead of $\cdot$ and $e$ in an abelian group -- such relations can be expressed as a triangular matrix multiplied by a vector of the element in $A$ equal to the zero vector.}. 

It is well known that every finite abelian group \( G \) is isomorphic to a direct product
of cyclic groups 
$
G_1 \times G_2 \times \cdots \times G_t,
$
where each \( G_i \) is a cyclic group of order \(m_j\ge 2 \). If \( a_i \) generates the cyclic group \( G_i \), \( i = 1, 2, \ldots, t \), i.e., $G_i=\langle a_i\rangle$,
then the elements \( a_1, a_2, \ldots, a_t \) are called a {\it basis} of \( G \). 
Since each $\langle a_i\rangle$ is isomorphic to $\ZZ_{m_i}$ for $m_i=|\langle a_i\rangle|$, it follows that $G$ is isomorphic to $\ZZ_{m_1}\times \cdots\times \ZZ_{m_t}$. 

The following lemma is straightforward to prove.

\begin{lemma}\label{TrivR}
    Let $(G,\cdot)$ be a group, and $A$ be any set such that $|A|=|G|$. Let $\Pi:G\to A$ be a bijective function. Then $(A,\circ)$, where the operation $\circ$ is defined as 
    $a\circ b=\Pi(\Pi^{-1}(a)\cdot \Pi^{-1}(b))$ is a group that is isomorphic to $(G,\cdot)$. 
\end{lemma}

\subsection{The Monomial Abelian Group}
Let $n\in\mathbb{N}$, $n\ge 2$, and $K=(\kappa_1,\ldots,\kappa_t)\in \mathbb{N}^t$ such that $\kappa_i\ge 2$ and $\kappa_1\kappa_2\cdots \kappa_t=n$. Let $0\le \ell_{i,j}\le \kappa_j-1$ for $i=2,\ldots,t$ and $j=1,\ldots,i-1$, and define $L=(\ell_{i,j})$. We define 
   $$\Gamma(K,L)=\langle x_1,x_2,\ldots,x_t\ |\ x_i^{\kappa_i}= x_1^{\ell_{i,1}}\cdots x_{i-1}^{\ell_{i,i-1}},i=2,\ldots,t\ ;\ x_1^{\kappa_1}=1\rangle$$
as the set of all monomials over the variables $x_1,x_2,\ldots,x_t$ of the form $x_1^{j_1}x_2^{j_2}\cdots x_t^{j_t}$, where $0\le j_i\le \kappa_i-1$ for all $i\in [t]$, with multiplication modulo $x_i^{\kappa_i}- x_1^{\ell_{i,1}}\cdots x_{i-1}^{\ell_{i,i-1}}$ for all $i=2,\ldots,t$ and $ x_1^{\kappa_1}-1$. 
%Notice that the relations $x_i^{\kappa_i}= x_1^{\ell_{i,1}}\cdots x_{i-1}^{\ell_{i,i-1}},i=2,\ldots,t\ ;\ x_1^{\kappa_1}=1$ are triangular relations. 

\vspace{.1in}
\noindent
{\bf Example:} Let $n=36$, $K=(4,3,3)$, and 
$L=(\ell_{2,1},\ell_{3,1},\ell_{3,2})=(3,2,1).$ Then
$$\Gamma(K,L)=\langle x_1,x_2,x_3\ |\ x_1^4=1;\ x_2^3=x_1^3;\ x_3^3=x_1^2x_2\rangle.$$
The multiplication of $x_1^2x_2^2x_3^2$ with $x_1^3x_2x_3^2$ is
$$(x_1^2x_2^2x_3^2)(x_1^3x_2x_3^2)=x_1^5x_2^3x_3^4=x_1x_1^3(x_1^2x_2x_3)=x_1^2x_2x_3.$$

The following lemma, proved in~\cite{bshouty2025}, shows that for any $K$ and $L$, the set $\Gamma(K,L)$ is an abelian group. 

\begin{lemma}\label{GamisAb}\cite{bshouty2025}
   Let $\kappa_1,\ldots,\kappa_t\in \mathbb{N}$ such that $\kappa_i\ge 2$ and $\kappa_1\kappa_2\cdots \kappa_t=n$. Let $0\le \ell_{i,j}\le \kappa_j-1$ for $i=2,\ldots,t$ and $j=1,\ldots,i-1$. Then, for $K=(\kappa_i)$ and $L=(\ell_{i,j})$,    
   \begin{enumerate}
       \item 
   $\Gamma(K,L)$
   is an abelian group of size $|\Gamma(K,L)|=n$.
   \item\label{GamisAb2} The multiplication of two elements and finding the inverse of an element in $\Gamma(K,L)$ can be computed in time $\tilde O(\log^2 n)$. 
   \end{enumerate}  
\end{lemma}
We refer to $\Gamma(K,L)$ the {\it monomial group generated by $K$ and $L$}. In this abelian group, $X=\{x_1,\ldots,x_t\}$ is a set of generators with the triangular relations $x_1^{\kappa_1}=1$ and $x_i^{\kappa_i}= x_1^{\ell_{i,1}}\cdots x_{i-1}^{\ell_{i,i-1}}$ for all $i=2,\ldots,t$. 

The following result follows from Kannan and Bachem~\cite{KannanB79} and Item~\ref{GamisAb2} in Lemma~\ref{GamisAb}. See details in~\cite{bshouty2025}. 
\begin{lemma}\label{ConBasis}
    There is a deterministic poly$(t,\log n)$ time algorithm that, for any monomial abelian group $\Gamma(K,L)$, returns integers $m_1|m_2|\cdots|m_r$ such that $\Gamma(K,L)$ is isomorphic to $\ZZ_{m_1}\times \ZZ_{m_2}\times \cdots\times\ZZ_{m_r}$. The algorithm also provides a basis $y_1,\ldots,y_r$ for $\Gamma(K,L)$.
\end{lemma}

There is also a randomized algorithm~\cite{KaltofenV05} with better time complexity, but the above lemma is sufficient for our purposes. 

\subsection{Algorithms in Abelian Group}
In~\cite{bshouty2025}, Bshouty proves the following result. 

\begin{lemma}\label{AlgGabelian21}
    Given a Cayley multiplication table of an abelian group $G$, where each entry in the table can be accessed in constant time, there exists a deterministic algorithm {\bf Generators} (in the FS-model) that runs in time $\tilde O(|G|)$ and a randomized algorithm {\bf Random-Generators} (in the PS-model) that runs in time $\tilde O(\sqrt{|G|})$. These algorithms provide 
    \begin{enumerate}
        \item A set $A=\{a_1,\ldots,a_t\}$ of generators for $G$ of size at most $\log |G|$.
        \item Integers $K=(k_i)$, $i\in [t]$, where $k_i\ge 2$, $\prod_{i\in[t]}k_i=|G|$ along with integers $L=(\lambda_{i,j})$, where $i\in[t],j\in [i-1]$ and $0\le \lambda_{i,j}\le k_i-1$. 
        \item $K$ and $L$ satisfy 
        \begin{eqnarray*}  a_i^{k_i}=a_1^{\lambda_{i,1}}\cdots a_{i-1}^{\lambda_{i,i-1}}.
        \end{eqnarray*}
    \end{enumerate}
Additionally,
    $G$ is isomorphic to $\Gamma(K,L)$ via the isomorphism
    $$\Psi(x_1^{j_1}\cdots x_t^{j_t})=a_1^{j_1}\cdots a_t^{j_t},$$ for every $0\le j_i\le k_i-1$ and $i\in [t]$.
\end{lemma}

\color{black}
\section{A One-Sided Tester for Abelian Group}\label{Sec4}

In this section, we prove the following two theorems in the FS-model.

{\bf Why not the PS-model?}
The one-sided tester presented in this section is stated only in the FS-model. In the PS-model, the size of the group is unknown and the algorithm only receives uniformly random elements from an oracle. Since a one-sided tester must accept every abelian group with probability $1$, its behavior on Accept-instances must be correct for every possible sequence of oracle outputs. However, after any finite sequence of samples, there is always a positive probability that the oracle continues returning only previously seen elements. Therefore, no deterministic procedure can certify that all elements of the group have been observed and hence cannot determine $|G|$ with certainty. 

\begin{theorem}\label{kkdd}
    There exists a one-sided tester for abelian groups that runs in time $\tilde O(|G|+1/\epsilon)$. 
\end{theorem}

\begin{theorem}\label{kkdd2}
    Let ${\cal G}$ be a subclass of abelian groups that is closed under isomorphism. If there exists a deterministic algorithm that can determine, in time $T$, whether any abelian group of the form $\ZZ_{m_1}\times \cdots\times \ZZ_{m_r}$ belongs to ${\cal G}$, then there exists a one-sided tester for ${\cal G}$ in the PS-model that runs in time $T+\tilde O({|G|}+1/\epsilon)$ and makes $\tilde O({|G|}+1/\epsilon)$ queries.

    In particular, there exist testers that run in time $\tilde O({|G|}+1/\epsilon)$ for the classes: abelian groups of rank at most $k$, abelian $p$-groups, vector spaces over $\ZZ_p$ and cyclic groups.
\end{theorem}

We first prove Theorem~\ref{kkdd}.

Given a binary operation $*:G^2\to G$, where each entry in the table can be accessed in constant time, we aim to test whether $(G,*)$ is an abelian group or $\epsilon$-far from being an abelian group. 

Consider the tester in Algorithm~\ref{AGOST}. The tester runs the deterministic Algorithm {\bf Generators} from Lemma~\ref{AlgGabelian21}. By Lemma~\ref{AlgGabelian21}, if $G$ is an abelian group, then the algorithm provides:
\begin{enumerate}
        \item A set $A=\{a_1,\ldots,a_t\}$ of generators for $G$ of size at most $\log |G|$.
        \item Integers $K=(k_i)$, where $i\in [t]$ such that $k_i\ge 2$, $\prod_{i\in[t]}k_i=|G|$, along with integers $L=(\lambda_{i,j})$, where $i\in[t],j\in [i-1]$ and $0\le \lambda_{i,j}\le k_i-1$. 
\end{enumerate}
Additionally, $G$ is isomorphic to $\Gamma(K,L)$ via the isomorphism
    $$\Psi(x_1^{j_1}\cdots x_t^{j_t})=a_1^{j_1}*a_2^{j_2}*\cdots *a_t^{j_t}$$ for every $0\le j_i\le k_i-1$ and $i\in [t]$.

When $G$ is not an abelian group, we cannot guarantee any of the conditions above. 
If the algorithm {\bf Generators} runs more than time $\tilde O(|G|)$ time, or if $K=(k_i)$ and $L=(\lambda_{i,j})$ do not satisfy item~\ref{OS2}, or if $|A|>\log |G|$, then $(G,*)$ is not an abelian group, and the tester rejects. See steps~\ref{LargeG} and~\ref{AKL}.

Otherwise, the tester computes $\Psi(z)$ for all $z\in \Gamma(K,L)$, build a table $\Psi:\Gamma(K,L)\to G$, and verifies whether $\Psi$ is bijective function. If not, the tester rejects. See step~\ref{PsiNb}.

Here, the tester computes $\Psi$ as follows
$$\Psi(x_1^{j_1}\cdots x_t^{j_t})=a_1^{j_1}*(a_2^{j_2}*(\cdots (a_{t-1}^{j_{t-1}}*a_t^{j_t}))\cdots)$$ where the exponent is handled recursively as follows
    $$a^m = 
\begin{cases} 
    (a^{m/2})*(a^{m/2}) & \text{if } m \text{ is even}, \\
    a * (a^{(m-1)/2}*a^{(m-1)/2}) & \text{if } m \text{ is odd}.
\end{cases}$$

If $\Psi$ is a bijective function, then the tester chooses $O(1/\epsilon)$ pairs $(x,y)\in \Gamma(K,L)^2$ uniformly at random and, using the table, verifies whether $$\Psi(xy)=\Psi(x)*\Psi(y).$$ If this condition holds for all tested pairs, the tester accepts; otherwise, it rejects.

\begin{algorithm}
\caption{One-Sided Tester for Abelian Group}\label{AGOST}
\begin{algorithmic}[1]
    \STATE Run the deterministic algorithm \textbf{Generators}$(G,*)$. 
    \STATE\label{LargeG} {\bf if} {the algorithm runs more than $\tilde{O}(|G|)$} steps {\bf then} Reject.
        \STATE {\bf else} By Lemma~\ref{AlgGabelian21}  obtain:
    \begin{enumerate}
        \item\label{OS1} A set $A = \{a_1, \ldots, a_t\}$ of size at most $\log |G|$.
        \item\label{OS2} Integers $K = (k_i)$ for $i \in [t]$, where $k_i \ge 2$ and $\prod_{i\in [t]}k_i=|G|$ along with integers $L = (\lambda_{i,j})$ for $i \in [t], j \in [i-1]$, where $0 \le \lambda_{i,j} \le k_i - 1$.
    \end{enumerate}
    
    \STATE\label{AKL} {\bf if} {the output $(A,K, L)$ does not satisfy the above conditions \ref{OS1}-\ref{OS2}} {\bf then} Reject
        
    \STATE Construct a table of values for the map $\Psi: \Gamma(K, L) \to G$, where for every $j\in \prod_{i\in [t]}[k_i]$: 
    \[
    \Psi(x_1^{j_1} \cdots x_t^{j_t}) = a_1^{j_1}*(a_2^{j_2}*(\cdots (a_{t-1}^{j_{t-1}}*a_t^{j_t}))\cdots).
    \] 
    \STATE\label{PsiNb} {\bf if} {$\Psi$ is not bijective} {\bf then} Reject.
    \STATE\label{Tfacirestar01}\label{FFor} {\bf Repeat} the following steps \( O(1/\epsilon) \) times:
        \STATE\ \ \  Choose $(x, y) \in \Gamma(K,L)^2$ uniformly at random.        \STATE\label{Tfacirestar02}\ \ \  \textbf{if} \[\Psi(xy)\neq \Psi(x) * \Psi(y)   \]\STATE \label{LReject}\textbf{\hspace{.09in} then} Reject
    \STATE Accept.
\end{algorithmic}
\end{algorithm}

We now prove its correctness.

\noindent
{\bf Completeness:} If $(G,*)$ is an abelian group, then by Lemma~\ref{AlgGabelian21}, conditions~\ref{OS1} and~\ref{OS2} in the tester hold, and $\Gamma(K,L)$ is isomorphic to $G$ via the isomorphism $\Psi$. Therefore, $\Psi$ is bijective, and $\Psi(xy)=\Psi(x)*\Psi(y)$. Hence, the tester accepts. 

\vspace{.1in}
\noindent
{\bf Soundness:} Suppose $(G,*)$ is $\epsilon$-far from every abelian group. If the tester does not reject in step~\ref{LargeG} or step~\ref{AKL}, then, by Lemma~\ref{GamisAb}, $\Gamma(K,L)$ is an abelian group of size $|\Gamma(K,L)|=|G|$. 
If the tester does not reject in step~\ref{PsiNb}, then $\Psi$ is bijective. By Lemma~\ref{TrivR}, the magma $(G,\circ)$ where $a\circ b=\Psi(\Psi^{-1}(a)\Psi^{-1}(b))$ is an abelian group isomorphic to $\Gamma(K,L)$. 
Since $(G,*)$ is $\epsilon$-far from any abelian group, it is also $\epsilon$-far from $(G,\circ)$. Therefore, since $\Phi$ is bijective, we have
\begin{eqnarray*}
    \Pr_{x,y\sim \Gamma(K,L)}[\Psi(xy)\not=\Psi(x)*\Psi(y)]&=&\Pr_{a,b\sim G}[\Psi(\Psi^{-1}(a)\Psi^{-1}(b))\not=a*b]\\
&=&\Pr_{a,b\sim G}[a\circ b\not= a*b]\ge \epsilon.
\end{eqnarray*}

Thus, with probability at least
$1-(1-\epsilon)^{O(1/\epsilon)}\ge 2/3,$
the tester rejects in step~\ref{LReject}.  

This completes the proof of the tester's correctness.

\vspace{.1in}
\noindent
{\bf Complexity}: By Lemma~\ref{AlgGabelian21}, the algorithm {\bf Generators} runs in time $\tilde O(|G|)$. Since $j\in \prod_{i\in[t]}[k_i]$ and $\prod_{i\in [t]}k_i=|G|$ computing $a_1^{j_1}*(a_2^{j_2}*(\cdots (a_{t-1}^{j_{t-1}}*a_t^{j_t}))\cdots)$ takes at most time 
$$(t-1)+\sum_{i\in [t]}\lceil \log k_i\rceil \le 2t+\log|G|\le 3\log|G|.$$
Thus, building the table takes time $\tilde O(|G|)$. 
By Lemma~\ref{GamisAb}, computing the product $xy$ takes time $O(\log^2|G|)$. 
Hence, this final test takes time $O(\log^2|G|/\epsilon)=\tilde O(1/\epsilon)$. Therefore, the time complexity of the tester is $\tilde O(|G|+1/\epsilon)$.

This completes the proof of Theorem~\ref{kkdd}. 

For testing subclasses of abelian groups ${\cal G}$ that are closed under isomorphism, Theorem~\ref{kkdd2}, we modify the procedure as follows: After step~\ref{AKL}, we run the deterministic algorithm from Lemma~\ref{ConBasis}, which returns integers $m_1 | m_2 | \cdots | m_r$ such that $\Gamma(K,L)$ is isomorphic to $G':=\ZZ_{m_1}\times \ZZ_{m_2}\times \cdots\times\ZZ_{m_r}$. 
We then run the deterministic algorithm to verify whether $G'\in {\cal G}$. If not, then we reject.  

The same analysis applies, since by Lemma~\ref{TrivR}, $(G,\circ)$ is isomorphic to $\Gamma(K,L)$. 

\section{A Tester for Abelian Groups}
 \label{ATAG}
In this section, we prove
\begin{theorem}\label{TestAbelianGroup}
    There exists a tester for abelian groups that runs in time $\tilde O(\sqrt{|G|}+1/\epsilon)$. 
\end{theorem}
\subsection{The Tester}
Consider the following tester.

\begin{algorithm}[H]
\caption{Tester For Abelian Groups}\label{AlgTest}
\begin{algorithmic}[1]
    \STATE Run the algorithm \textbf{Random-Generators}$(G,*)$.
    \STATE\label{LargeGR} {\bf if} {the algorithm runs in time more than $\tilde{O}(\sqrt{|G|})$} {\bf then} Reject.
        \STATE {\bf else} By Lemma~\ref{AlgGabelian21} we obtain
    \begin{enumerate}
        \item\label{OS1R} A set $A = \{a_1, \ldots, a_t\}$ of size at most $\log |G|$.
        \item\label{OS2R} Integers $K = (k_i)$ for $i \in [t]$, where $k_i \ge 2$, $\prod_{i\in [t]}k_i=|G|$ along with integers $L = (\lambda_{i,j})$ for $i \in [t], j \in [i-1]$, where $0 \le \lambda_{i,j} \le k_i - 1$.
    \end{enumerate}
    
    \STATE\label{AKLR} {\bf if} {the output $(A, K, L)$ does not satisfy the above conditions \ref{OS1R}-\ref{OS2R}} {\bf then} Reject
    \STATE Define $\Psi: \Gamma(K, L) \to G$ where for every $j\in \prod_{i\in [t]}[k_i]$ 
    \[
    \Psi(x_1^{j_1} \cdots x_t^{j_t}) = a_1^{j_1}*(a_2^{j_2}*(\cdots (a_{t-1}^{j_{t-1}}*a_t^{j_t}))\cdots).
    \] 
    \STATE\label{TKOLCK01} Choose  $m=O(\sqrt{|G|})$ elements $x^{(1)},x^{(2)},\ldots,x^{(m)} \in \Gamma(K,L)$ uniformly at random.
    \STATE\label{TKOLCK02} \textbf{If} $(\exists x^{(k)}\not=x^{(j)})\ \Psi(x^{(k)})=\Psi(x^{(j)})$ \textbf{then} Reject.
    \STATE\label{Tfacirestar011} {\bf Repeat} the following steps \( O(1/\epsilon) \) times:
        \STATE\ \ \  Choose $(x, y) \in \Gamma(K,L)^2$ uniformly at random.        \STATE\label{Tfacirestar021}\ \ \  \textbf{if} \[\Psi(xy)\neq \Psi(x) * \Psi(y)   \]\STATE \textbf{\hspace{.09in} then} Reject
    \STATE\label{acceptT} Accept
\end{algorithmic}
\end{algorithm}

Consider the tester in Algorithm~\ref{AlgTest}. The tester runs the randomized Algorithm {\bf Random-Generators}\footnote{In the PS-model, the size $|G|$ is not given as input; however, the procedure {\bf Random-Generators} returns $|G|$ as part of its output (Lemma~\ref{AlgGabelian21}).} from Lemma~\ref{AlgGabelian21}. By Lemma~\ref{AlgGabelian21}, if $G$ is an abelian group, then, with high probability, it provides:
\begin{enumerate}
        \item A set $A=\{a_1,\ldots,a_t\}$ of generators for $G$ of size at most $\log |G|$.
        \item Integers $K=(k_i)$, where $i\in [t]$ such that $k_i\ge 2$, $\prod_{i\in[t]}k_i=|G|$, along with integers $L=(\lambda_{i,j})$, where $i\in[t],j\in [i-1]$ and $0\le \lambda_{i,j}\le k_i-1$. 
\end{enumerate}
Additionally, $G$ is isomorphic to $\Gamma(K,L)$ via the isomorphism
    $$\Psi(x_1^{j_1}\cdots x_t^{j_t})=a_1^{j_1}*a_2^{j_2}*\cdots *a_t^{j_t}$$ for every $0\le j_i\le k_i-1$ and $i\in [t]$.

When $G$ is not an abelian group, we cannot guarantee any of the conditions above. If the algorithm {\bf Random-Generators} runs more than $\tilde O(\sqrt{|G|})$ time, or if $K=(k_i)$ and $L=(\lambda_{i,j})$ does not satisfy item~\ref{OS2R} or if $|A|>\log |G|$, then $(G,*)$ is not an abelian group, and the tester rejects. See steps~\ref{LargeGR} and~\ref{AKLR}.

Now, since we require the time complexity to be at most $\tilde O(\sqrt{|G|})$, the tester cannot compute $\Psi(z)$ for all $z\in \Gamma(K,L)$, and build a table $\Psi:\Gamma(K,L)\to G$, and verify whether $\Psi$ is a bijective function.

Instead, the tester performs the following two tests: 
\begin{enumerate}
    \item The first test verifies whether $\Psi$ is $\Theta(1)$-close to being a bijective map (see steps~\ref{TKOLCK01}-\ref{TKOLCK02}).
    \item The second test verifies whether $\Psi$ is $\Theta(\epsilon)$-close to being an isomorphism (see step~\ref{Tfacirestar02}). 
\end{enumerate}

The analysis in the following subsection shows that those two tests are sufficient. 

\subsection{Proof of Correctness and Complexity}
We now proceed to prove the correctness of the tester.

\vspace{.2in}

\noindent
{\bf Completeness:}
If $(G,*)$ is an abelian group, then, with high probability, {\bf Random-Generators} procedure returns a valid set of generators $A$, along with $K$ and $L$, that satisfy items~\ref{OS1R} and~\ref{OS2R}. Thus, the tester does not reject in steps~\ref{LargeGR} and~\ref{AKLR}. Consequently, $\Psi$ is bijective, and $G$ is isomorphic to $\Gamma(K,L)$ via $\Psi:\Gamma(K,L)\to G$. Therefore, the tester does not reject in step~\ref{TKOLCK02} and~\ref{Tfacirestar02}, and with high probability, it accepts. 

\vspace{.2in}

\noindent
{\bf Soundness:} Suppose $(G,*)$ is $\epsilon$-far from any abelian group.

If the tester does not reject in steps~\ref{LargeGR} or~\ref{AKLR}, then, by Lemma~\ref{GamisAb}, $\Gamma(K,L)$ is a group of size $\prod_{i\in[t]}k_i=|G|$. Thus
$$|\Gamma(K,L)|=|G|.$$

We now prove the following.
\begin{lemma}\label{OROR}
    Let $x_1,x_2$ be distinct elements in $\Gamma(K,L)$. If $\Psi(x_1)=\Psi(x_2)$, then for any $y\in \Gamma(K,L)$, at least one of the following holds:
    \begin{enumerate}
        \item\label{OROR1} $\Psi(x_1y)\not=\Psi(x_1)*\Psi(y)$.
        \item\label{OROR2} $\Psi(x_2y)\not=\Psi(x_2)*\Psi(y)$.
        \item\label{OROR3} $\Psi(x_1y)=\Psi(x_2y)$.
    \end{enumerate}
\end{lemma}
\begin{proof}
    If $\Psi(x_1y)\not=\Psi(x_2y)$ then since $\Psi(x_1)*\Psi(y)=\Psi(x_2)*\Psi(y)$, we must have either $\Psi(x_1y)\not=\Psi(x_1)*\Psi(y)$ or $\Psi(x_2y)\not=\Psi(x_2)*\Psi(y)$. 
\end{proof}
Next, we partition $\Gamma:=\Gamma(K,L)$ as $\Gamma=\Gamma_1\cup \Gamma_*$ where
$$\Gamma_1=\{x\in \Gamma|(\forall y\not=x)\ \Psi(x)\not= \Psi(y)\},\ \ \ \Gamma_*=\Gamma\backslash \Gamma_1=\{x\in\Gamma|(\exists y\not=x)\ \Psi(x)=\Psi(y)\}.$$

By the birthday paradox (see also~\cite{Amir0R24} subsection 1.2.3), we obtain the following result.
\begin{lemma}\label{KscK}
    If $|\Gamma_*|\ge |\Gamma|/8$, then for $m=O(\sqrt{|G|})$ elements  $x^{(1)},\ldots,x^{(m)}\in\Gamma$ chosen uniformly at random, with probability at least $5/6$, there is a pair $x^{(k)}\not =x^{(j)}$ such that $\Psi(x^{(k)})=\Psi(x^{(j)})$.
\end{lemma}

Thus, if $|\Gamma_*|\ge |\Gamma|/8$, then $(G,*)$ is not an abelian group, and with high probability, the algorithm rejects in steps~\ref{TKOLCK01}-\ref{TKOLCK02}. Therefore, we may assume that
\begin{eqnarray}\label{KSK8}
    |\Gamma_*|< |\Gamma|/8.
\end{eqnarray}

We now distinguish between two cases

\noindent
{\bf Case I.} $|\Gamma_1|\le (1-\epsilon/8)|\Gamma|$.

Then $|\Gamma_*|=|\Gamma|-|\Gamma_1|\ge (\epsilon/8)|\Gamma|$. We aim to partition \( \Gamma_* \) into pairs \(\{x, y\}\) such that \(x \neq y\), \(\Psi(x) = \Psi(y)\), and all the pairs are disjoint. While not every element in \( \Gamma_* \) will be included in a pair, we strive to create as many pairs as possible.

To achieve this, consider the relation $x\sim y$ if \(\Psi(x) = \Psi(y)\), and partition $\Gamma_*=\Gamma_*^{(1)}\cup \Gamma_*^{(2)}\cup\cdots \cup \Gamma_*^{(w)}$ according to this relation (i.e., the equivalence classes). By definition of $\Gamma_*$, each $\Gamma_*^{(i)}$ contains at least two elements with the same image under $\Psi$. From each $\Gamma_*^{(i)}$, we can create $\lfloor |\Gamma_*^{(i)}|/2\rfloor$ disjoint pairs. 
Therefore, we can create at least
$$\sum_{i=1}^w\left\lfloor\frac{|\Gamma_*^{(i)}|}{2}\right\rfloor \ge \sum_{i=1}^w\frac{|\Gamma_*^{(i)}|}{3}\ge \frac{|\Gamma_*|}{3}\ge \frac{\epsilon}{24}|\Gamma|$$
disjoint pairs. 

Let $\{x_1,x_2\}$ be one such pair, where $x_1\not=x_2$ and $\Psi(x_1)=\Psi(x_2)$. By Lemma~\ref{OROR}, for every $y\in \Gamma$, one of the items~\ref{OROR1}-\ref{OROR3} in Lemma~\ref{OROR} occurs. 
If item~\ref{OROR3} occurs for more than $|\Gamma|/2$ of the $y$ values, then $\Psi(x_1y)=\Psi(x_2y)$ for at least $|\Gamma|/2$ ordered pairs $(x_1y,x_2y)\in \Gamma\times \Gamma$. Since $\Gamma$ is a group, it follows that $x_1y\not=x_2y$ and $x_1y\not=x_1y'$ for $y\not=y'$. Consequently, at least $|\Gamma|/2$ values of $y$ satisfy $x_1y\in \Gamma_*$, implying that $|\Gamma_*|\ge |\Gamma|/2$, which contradicts (\ref{KSK8}).

Thus, for at least $|\Gamma|/2$ of the $y$ values, either item~\ref{OROR1} in Lemma~\ref{OROR} occurs (where  $\Psi(x_1y)\not=\Psi(x_1)*\Psi(y)$) or item~\ref{OROR2} occurs (where $\Psi(x_2y)\not=\Psi(x_2)*\Psi(y)$).

Since this holds for each pair $\{x,y\}$, and we have at least $(\epsilon/24)|\Gamma|$ such disjoint pairs, it follows that
$$\Pr_{x,y}[\Psi(xy)\not= \Psi(x)* \Psi(y)]\ge \frac{(\epsilon/48)|\Gamma|\cdot|\Gamma|/2}{|\Gamma|^2}\ge \frac{\epsilon}{96}.$$
Therefore, with high probability, the tester rejects in steps~\ref{Tfacirestar011}-\ref{Tfacirestar021}.

\noindent
{\bf Case II.} $|\Gamma_1|> (1-\epsilon/8)|\Gamma|$.

Let $\A:\Gamma\to G$ be any bijective function such that $\A(x)=\Psi(x)$ if $x\in \Gamma_1$. Define the binary operation $\bullet$ on $G$ as follows: 
$$a\bullet b=\A(\A^{-1}(a)\A^{-1}(b)).$$

By Lemma~\ref{TrivR}, we have the following:
\begin{lemma}\label{Gbullet}
    $(G,\bullet)$ is an abelian group.
\end{lemma}

We now prove the following result.
\begin{lemma}\label{KKKK}
    Let $Y$ be the event $[x\in \Gamma_1, y\in \Gamma_1, xy\in \Gamma_1]$. Then
    $$\Pr_{x,y}[Y]\ge 1-\frac{3}{8}\epsilon.$$
\end{lemma}
\begin{proof}
    We have 
    $$ \Pr_x[x\not\in \Gamma_1]=\Pr_y[y\not\in \Gamma_1]=\frac{|\Gamma_*|}{|\Gamma|}=\frac{|\Gamma|-|\Gamma_1|}{|\Gamma|}\le \frac{\epsilon}{8}.$$
    Since every group element appears exactly $|\Gamma|$ times in the Cayley table of $\Gamma$, it follows that 
    $$\Pr_{x,y}[xy\not\in \Gamma_1]=\frac{|\Gamma_*|\cdot|\Gamma|}{|\Gamma|^2}\le \frac{\epsilon}{8},$$
    and the result follows from the union bound.
     \end{proof}
  If $Y$ occurs, then $\Psi(xy)=\A(xy)$, and
$$\Psi(xy)=\A(xy)=\A(\A^{-1}(\A(x))\A^{-1}(\A(y)))=\A(x)\bullet\A(y).$$  
 Since $\A$ is bijective, $(G,\bullet)$ is abelian group, $(G,*)$ is $\epsilon$-far from any abelian group, and by Lemma~\ref{KKKK},
\begin{eqnarray}
    \Pr_{x,y}[\Psi(xy)\not=\Psi(x)*\Psi(y)]
    &\ge & \Pr_{x,y}[\Psi(xy)\not=\Psi(x)*\Psi(y)|Y]\Pr[Y]\nonumber\\
    &=&\Pr_{x,y}[\A(x)\bullet \A(y)\not=\A(x)*\A(y)|Y]\Pr[Y]\nonumber\\
    &\ge&(\Pr_{x,y}[\A(x)\bullet \A(y)\not=\A(x)*\A(y)]-\Pr[\neg Y])\Pr[Y]\nonumber\\
    &=&(\Pr_{a,b\sim G}[a\bullet b\not=a*b]-\Pr[\neg Y])\Pr[Y]\nonumber\\
    &\ge&(\epsilon-3\epsilon/8)(1-3\epsilon/8)\ge \frac{\epsilon}{3}.\nonumber
\end{eqnarray}
Thus, with high probability, the tester rejects in steps~\ref{Tfacirestar011}-\ref{Tfacirestar021}.
This completes the proof of correctness for the tester.

\vspace{.2in}
\noindent
{\bf Time Complexity:}
The time complexity of the tester is as follows. By Lemma~\ref{AlgGabelian21}, time complexity of {\bf Random-Generators} is $\tilde O(\sqrt{|G|})$. 
Computing $\Psi(\alpha)$ requires time at most $$\sum_{i=1}^t\lceil\log k_i\rceil\le 2t+\log (k_1k_2\cdots k_t)\le 3\log |G|.$$
Therefore, the time complexity of steps~\ref{TKOLCK01}-\ref{TKOLCK02} is $\tilde O(\sqrt{|G|})$. By Lemma~\ref{GamisAb}, operations in $\Gamma$ take time $\tilde O(\log^2 |G|)$, so the time complexity of steps~\ref{Tfacirestar011}-\ref{Tfacirestar021} is $\tilde O(1/\epsilon)$. Thus, the overall time complexity of the tester is $\tilde O(\sqrt{|G|}+1/\epsilon)$.

\section{Testing Subclasses of Abelian Groups}\label{TSAG}
In this section, we prove.
\begin{theorem}
    Let ${\cal G}$ be a subclass of abelian groups that is closed under isomorphism. Suppose there exists an algorithm ${\cal A}$ that can decide if any abelian group of the form $G'=\ZZ_{m_1}\times \cdots \times\ZZ_{m_r}$ belongs to ${\cal G}$ in time $T$. Then, there exists a tester for ${\cal G}$ that runs in time $T+\tilde O(\sqrt{|G|}+1/\epsilon)$ and makes $\tilde O(\sqrt{|G|}+1/\epsilon)$ queries.
\end{theorem}
In particular, we have
\begin{theorem}
    There exist testers for abelian groups of rank at most $k$, abelian $p$-groups, vector spaces, and cyclic groups  over~$\ZZ_p$ that run in time $ O(\sqrt{|G|}+1/\epsilon)$.
\end{theorem}
\begin{proof}
    The tester follows the same structure as the tester for abelian groups in Algorithm~\ref{AlgTest}, with an additional two steps before step~\ref{TKOLCK01}. These additional steps are: 
    \begin{enumerate}
        \item\label{VS01} Finding a basis for $\Gamma:=\Gamma(K,L)$ to determine the abelian group $G'=\mathbb{Z}_{m_1}\times\cdots\times \mathbb{Z}_{m_r}$ that is isomorphic to $\Gamma$.
        \item\label{VS02} Verifying whether $G'\in {\cal G}$ using algorithm ${\cal A}$. If $G'\not\in {\cal G}$, then the tester rejects. 
    \end{enumerate}
    The basis can be found using the algorithm from Lemma~\ref{ConBasis}.
    
    If $\Gamma$ is isomorphic to $G'$ and $G'\in {\cal G}$, it follows that $\Gamma\in {\cal G}$. We now revisit the proof of Theorem~\ref{TestAbelianGroup} to show that if $\Gamma\in {\cal G}$, then the tester in Algorithm~\ref{AlgTest} is a tester for ${\cal G}$. 
    
\vspace{.2in}

    \noindent{\bf Completeness}: If $(G,*)$ is an abelian group in ${\cal G}$, then all the arguments from the completeness proof for the abelian group remain valid. Consequently, the algorithm does not reject in the steps~\ref{TKOLCK02} and~\ref{Tfacirestar02}. 
    
    Since $\Gamma\in {\cal G}$, the additional verification step~\ref{VS02} also does not reject, and the tester accepts.

\vspace{.2in}

    \noindent 
    {\bf Soundness}: Suppose $(G,*)$ is $\epsilon$-far from any abelian group in ${\cal G}$.

    First, Lemma~\ref{OROR} still applies. The sets $\Gamma_1$ and $\Gamma_*$ are defined in the same way. Consequently, Lemma~\ref{KscK} holds as well. If $|\Gamma_*|\ge |\Gamma|/8$, then $(G,*)$ is not an abelian group and, therefore, is not in~${\cal G}$. In that case, with high probability, the algorithm rejects in steps~\ref{TKOLCK01}-\ref{TKOLCK02}. Thus, we can assume 
$$|\Gamma_*|\le |\Gamma|/8.$$

Since Case I follows exactly as before, we proceed with Case II, where $|\Gamma_1|>(1-\epsilon/8)|\Gamma|$. 

Let $\A:\Gamma\to G$ be any bijective function such that $\A(x)=\Psi(x)$ for $x\in \Gamma_1$. Define a binary operation $\bullet$ on $G$ as follows: $$a\bullet b=\A(\A^{-1}(a)\A^{-1}(b)).$$

The following lemma replaces Lemma~\ref{Gbullet}. 
\begin{lemma}\label{FarFrom}
    $(G,\bullet)$ is an abelian group that is isomorphic to $\Gamma$.

    In particular, $(G,*)$ is $\epsilon$-far from $(G,\bullet)$.
\end{lemma}
\begin{proof}
The lemma follows from Lemma~\ref{TrivR}.

Since $\Gamma$ is in ${\cal G}$, we have that $(G,\bullet)$ is in ${\cal G}$. Since $(G,*)$ is $\epsilon$-far from any group in ${\cal G}$ we have that $(G,*)$ is $\epsilon$-far from $(G,\bullet)$.
\end{proof}

Now, Lemma~\ref{KKKK} and all the associated arguments also hold in this setting. 
The only difference is that we now derive 
$$\Pr_{a,b\sim G}[a\bullet b\not=a*b]\ge \epsilon$$
because, by Lemma~\ref{FarFrom}, $(G,*)$ is $\epsilon$-far from $(G,\bullet)$. 
\end{proof}

\section{Lower Bounds}\label{SecLB}

In this section, we present lower bounds for testing subclasses of abelian groups.

Bshouty~\cite{bshouty2025} proved the following two lower bounds.

\begin{lemma}\label{IsoKnown}
    Let ${\cal G}$ be the set of all groups that are isomorphic to either $H_1={\mathbb{Z}}_{2^2}^m$ or $H_2=\mathbb{Z}_{2^2}^{m-1}\times \mathbb{Z}_2^2$. Any algorithm that, with probability at least $2/3$, determines whether $G\in {\cal G}$ is isomorphic to $H_1$ or $H_2$ must access the elements of $G$ and its Cayley table at least $\Omega({|G|}^{1/4})$ times. 
\end{lemma}

\begin{lemma}\label{IsunKnown}
    Let ${\cal G}$ be the set of all groups that are isomorphic to either $H_1=\mathbb{Z}_{2}^{m-1}$ or $H_2=\mathbb{Z}_{2}^m$. Any algorithm that, for $G\in {\cal G}$ of unknown size, decides with probability at least $2/3$ whether $G$ is isomorphic to $H_1$ or $H_2$ must access an oracle that selects uniformly random elements of $G$ and access the Cayley table of $G$ at least $\Omega({|G|}^{1/2})$ times. 
\end{lemma}
In \cite{GallY13}, Gall and Yoshida proved the following result.
\begin{lemma}\label{TGD}
   For any two non-isomorphic groups $G$ and $H$, $G$ is $1/23$-far from $H$.
\end{lemma}

We now give the lower bounds for testers.
\begin{theorem}\label{lb03}
    Let ${\cal G}$ be the set of all abelian groups isomorphic to $H_1={\mathbb Z}_{2^2}^m$. Then,
    in the FS-model, any tester for ${\cal G}$ with $\epsilon<1/23$ must run in time $\Omega({|G|^{1/4}})$.
\end{theorem}
\begin{proof} Let ${\cal T}$ be a tester for ${\cal G}$ that runs in time $t$. Define the set of groups ${\cal G}'$ to be all groups that are isomorphic to either $H_1={\mathbb{Z}}_{2^2}^m$ or $H_2=\mathbb{Z}_{2^2}^{m-1}\times \mathbb{Z}_2^2$. 
Given $G\in {\cal G}'$, we run the tester ${\cal T}$ on $G$ with $\epsilon<1/23$. 
If $G$ is isomorphic to $H_1$, the tester accepts with probability at least 
$2/3$. If $G$ is isomorphic to $H_2$, then since, by Lemma~\ref{TGD}, $H_2$ is $1/23$-far from $G$, the tester rejects with probability at least $2/3$. Thus, the tester ${\cal T}$ can distinguish between $G\in {\cal G}'$ that are isomorphic to $H_1$ and those that are isomorphic to $H_2$. By Lemma~\ref{IsoKnown}, it must run in time $\Omega(|G|^{1/4})$. 
\end{proof}

We now prove the following lower bound.
\begin{theorem}\label{lb04}
    Let ${\cal G}$ be the set of all abelian groups of even rank. Then,
    in the PS-model, any tester for ${\cal G}$ with $\epsilon<1/23$ must run in time $\Omega(\sqrt{|G|})$.
\end{theorem}
\begin{proof}
    Suppose there exists a tester ${\cal T}$ for ${\cal G}$ that runs in time $T(|G|)$. 
    Consider the set of groups ${\cal G}'$ that are isomorphic to either $H_1'=\ZZ_2^{2m}$ or $H_2'=\ZZ_2^{2m-1}$. 
    Given $G$ that is either isomorphic to $H_1'$ or $H_2'$, we run the tester ${\cal T}$. If the tester accepts, then $G$ is isomorphic to $H_1'$. 
    This is because $H_1'$ is of even rank, while $H_2'$ is $1/23$-far from any group of size $2^{2m-1}$ with even rank. For the same reason, if the tester rejects, $G$ is isomorphic to $H_2'$.

    By Lemma~\ref{IsunKnown}, for $n=2^{2m}$, we have $T(n) =\Omega(\sqrt{n})$. 
\end{proof}

It follows from~\cite{Fischer24} that any tester for abelian groups or their subclasses must make at least $\Omega(1/\epsilon)$ queries. 
This also follows from the fact that if we take any group and modify an $\epsilon<1/46$ fraction of the entries of its Cayley table, chosen uniformly at random, to different values, we obtain a magma $H$ that is $\epsilon$-far from every group. 
At least  $\Omega(1/\epsilon)$ queries are required to identify one of these modified entries.

\bibliographystyle{plainurl}
\bibliography{TestingRef}

\end{document}